\documentclass[final]{siamltex}

\newcommand*\samethanks[1][\value{footnote}]{\footnotemark[#1]}
\author{Armeen Taeb\thanks{Dept.~of Electrical, Computer, and Energy Engineering, University of Colorado at Boulder.}
\and Arian Maleki\thanks{Dept.~of Statistics, Columbia University.} 
\and Christoph Studer\thanks{Dept.~of Electrical and Computer Engineering, Rice University.} 
\and Richard~G.~Baraniuk\samethanks}

\title{Maximin Analysis of Message Passing Algorithms for Recovering Block Sparse Signals}

\usepackage{placeins}
\usepackage{mathtools}
\usepackage{subfigure}
\usepackage{upgreek}
\usepackage{amsmath}
\usepackage{amssymb}
\usepackage{graphicx}
\usepackage{mathtools}
\usepackage{parskip}
\usepackage{color}
\usepackage{cite}

\begin{document}

\maketitle



\begin{abstract}
We consider the problem of recovering a block (or group) sparse signal from an underdetermined set of random linear measurements, which appear in compressed sensing applications such as radar and imaging. Recent results of Donoho, Johnstone, and Montanari have shown that approximate message passing (AMP) in combination with Stein's shrinkage outperforms group LASSO for large block sizes. 
In this paper, we prove that, for a fixed block size and in the strong undersampling regime (i.e.,  having very few measurements compared to the ambient dimension), AMP cannot improve upon group LASSO, thereby complementing the results of Donoho \emph{et al.}
\end{abstract}


\pagestyle{myheadings}
\thispagestyle{plain}

\section{Introduction}


The field of compressed sensing (CS) aims to recover a sparse signal from an undetermined systems of linear equations. Concretely, CS can be modeled as $\mathbf{y} = A\mathbf{x}$, where $\mathbf{y}$ is the $n$-dimensional measurement vector, $A$ is the (typically random) $n \times N$ measurement matrix, and $\mathbf{x}$ is an $N$ dimensional vector with at most~$k$ nonzero entries (often referred to as $k$-sparse vector). Our goal is to recover $\mathbf{x}$ from this undetermined system.

A large class of signals of interest exhibit additional structure known as \emph{block (or group) sparsity}, where the non-zero coefficients of the signal occur in clusters of size~$B$~\cite{YL06,Block2009}. Such block-sparse signals naturally appear in genomics, radar, and communication applications. There has been a considerable amount of research on theory and algorithms for recovering such signals~\cite{YL06,CH08,NR08, Bach08,MGB08,Block2009, LoPoTsGe09, MaCeWi05, ObWaJo11, duarte2011performance, lvgroup}. Perhaps the most popular recovery algorithm is \emph{group LASSO}~\cite{YL06}---corresponding theoretical work has shown that under which conditions this algorithm recovers the exact group sparse solution \cite{CH08,NR08, Bach08,MGB08,Block2009, LoPoTsGe09, MaCeWi05, ObWaJo11, duarte2011performance, lvgroup}. While these results enable a qualitative characterization of the recovery performance of group LASSO, they do not provide an accurate performance analysis. 

In order to arrive at a more accurate performance analysis of group LASSO, several authors have considered the asymptotic setting where $n, N \rightarrow \infty$, while their ratio $\delta = n/N$ is held constant \cite{DoTa08ArXiv, IntroAMP, stojnic2009block, CAMP, DONOHOGS}.  Under this setting, the references \cite{stojnic2009block, CAMP, DONOHOGS}  have shown that there exists a threshold on the normalized sparsity~$\rho = k/n$, below which group LASSO recovers the correct signal vector~${\bf x}$ with probability $1$ and fails otherwise. Such a  phase-transition (PT) analysis has led to the conclusion that group LASSO is sub-optimal, since there is a large gap between the PT of the information theoretic limit\footnote{The information theoretic limit was only derived for regular sparse signals (with block size 1) in \cite{WuVe10}. An extension of these results to block sparse signals with larger blocks is straightforward.} and that of  group LASSO (see Fig. \ref{Fig:groupLassopT}). 
\begin{figure}[!htbp]
  \centering
  	\includegraphics[width=.5\textwidth]{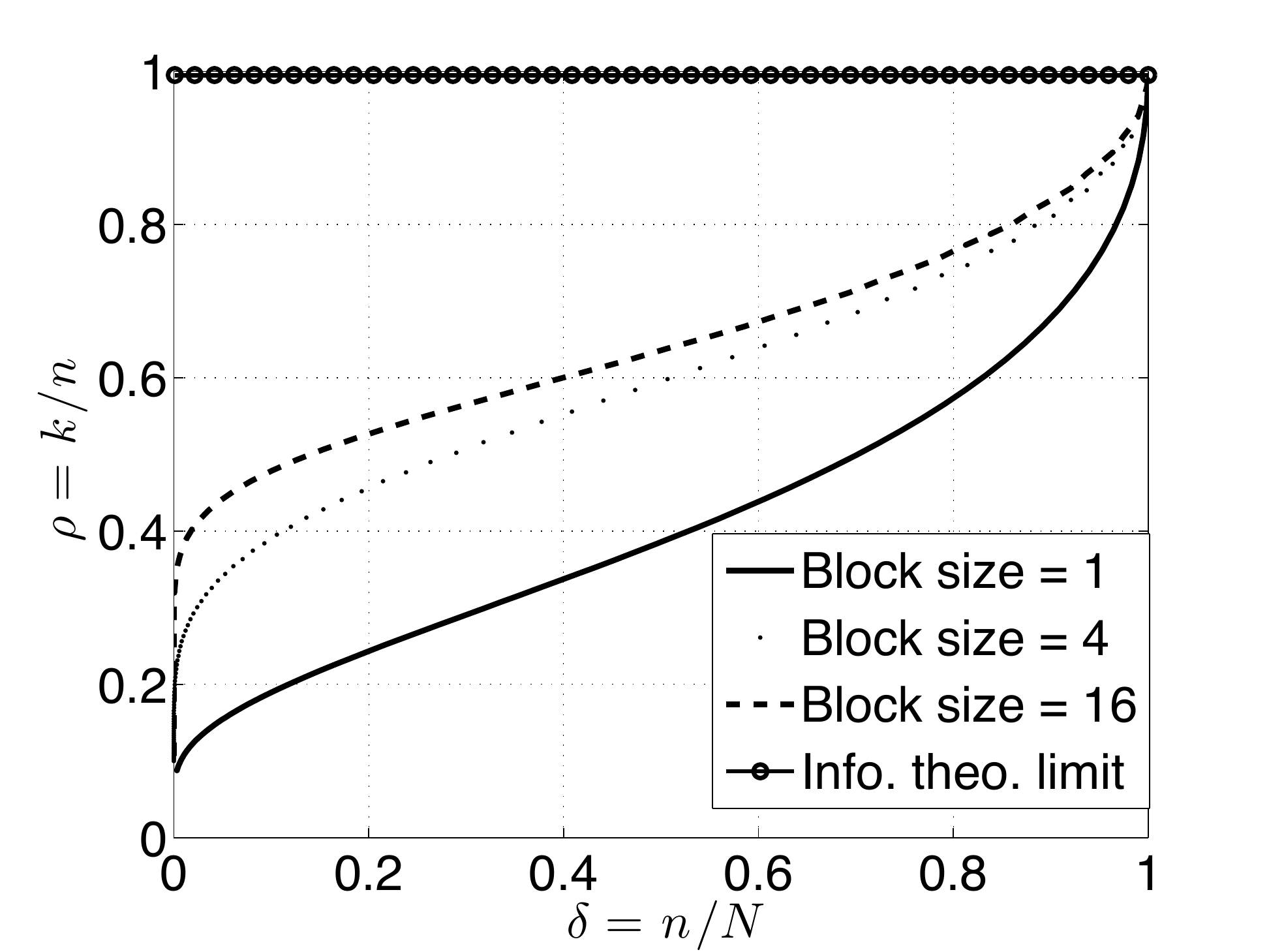}
	\vspace{-0.2cm}
  	\caption{Phase transition (PT) of group LASSO for various block sizes $B$. Evidently, there is a disparity between the phase transition of group LASSO and the information theoretic limit. }
	\vspace{-0.3cm}
     \label{Fig:groupLassopT}
\end{figure}

There has been recent effort in using \emph{approximate message passing} (AMP) to improve upon the performance of group LASSO.  Schniter, for example, has experimentally shown in \cite{Schnit} that AMP combined with expectation maximization can outperform group LASSO. Kamilov \emph{et al.}~have taken the first step toward a theoretical understanding of such algorithms \cite{Fletch}. More recently, Donoho, Johnstone, and Montanari \cite{DONOHOGS} have shown that AMP is able to outperform group LASSO if it employs Stein's shrinkage estimator. In fact, they demonstrate that for very large blocks sizes, i.e., for $B\rightarrow\infty$, the performance of this AMP variant is close to the information theoretic limit. However, in many applications, such as radar and communication systems, the group sizes are typically small and for fixed block sizes, Stein's estimator does not necessarily improve the performance of AMP. 
Consequently, the fundamental question remains whether AMP can outperform group LASSO for any block size. In this paper, we address this question in the high undersampling regime where $\delta\rightarrow0$.

In particular, we show that, for $\delta \rightarrow 0$, there is no nonlinear shrinkage function that allows AMP to outperform group LASSO. We emphasize that this result does not contradict that in \cite{DONOHOGS} as they considered a different limiting regime, i.e., where the block size $B$ approaches infinity. 
A combination of these two results enables us to conclude that, for strong undersampling (i.e., small values of $\delta$), AMP with Stein's estimator requires large block sizes in order to outperform group LASSO.

\section{Background}


\vspace{-0.15cm}
\subsection{Notation} \label{sec:notation}

\vspace{-0.2cm}

Lowercase boldface letters, such as ${\bf v}$, represent vectors and uppercase letters, such as $V$, represent matrices; lowercase letters, such as $v$ represent scalars. 
We analyze the recovery of a block (or group) sparse signal ${\bf x} \in \mathbb{R}^{N}$ with at most $k$ nonzero entries from the undersampled linear measurements ${\bf y} = A {\bf x}$, where $A \in \mathbb{R}^{n \times N}$ is i.i.d. zero-mean Gaussian with unit variance. We furthermore consider the asymptotic setting where $\delta = n/N$, $\rho = k/n$, and $N,n,k \to \infty$.

The notational conventions for block sparse signals are as follows. We assume that all the blocks have the same size, denoted by $B$. Extensions to signals with varying block sizes is straightforward.  The signal ${\bf x}$ is partitioned into $M$ blocks where clearly, $N = MB$. In the remainder of the paper, we will denote ${\bf x}_B$ as a particular block.  Suppose that the elements of ${\bf x}_B$ are drawn from a given distribution $F({\bf x}_B)  = (1- \epsilon) \delta_{0}(\|{\bf x}_B\|_2) + \epsilon G({\bf x}_B)$, where $\epsilon = \rho\delta$, and $\delta_0$ is the Dirac delta function; $G$ is a probability distribution that is typically unknown in practice.

The block soft-thresholding function used in this paper is defined as follows~\cite{DONOHOGS}:
\begin{equation}
\eta^{\text{soft}}({\bf y} ;\tau) = \frac{{\bf y}_B}{\left\| {\bf y}_B\right\|_2} (\left\| {\bf y}_B\right\|_2-\tau)_{+} .
\label{eqn:softThresh}
\end{equation}
Here, $(z)_{+} = \max(z,0)$ and $\eta^{\text{soft}}({\bf y}_B ;\tau)$ sets its argument ${\bf y}_B$ to zero if $\left \|{\bf y}_B \right \|_2 \leq \tau$, and shrinks the vector ${\bf y}_B$ towards the origin by $\tau$, otherwise.


\subsection{Group LASSO and approximate message passing (AMP)} \label{sec:amp}

A decade of research in sparse recovery has produced a plethora of algorithms for recovering block sparse signals from random linear measurements. Two popular algorithms are group LASSO and AMP. Group LASSO searches for a vector~$\mathbf{x}$ that minimizes the cost function, ${\bf x}^{\rm L} \triangleq \arg \min_{{\bf x}} \{   \sum_{B=1}^M \|{\bf x}_B\|_2 \colon {\bf y}=A {\bf x} \}$. 
AMP, on the other hand, is an iterative algorithm to recover the solution vector ${\bf x}$. Concretely, by initializing ${\bf x}^0 = \mathbf{0}$ and ${\bf z}^0 = \mathbf{0}$, AMP iteratively performs the following steps:
\begin{align}
\label{eqn:AMPiter}
{\bf x}^{t+1} = \eta({\bf x}^t+ A^* {\bf x}^t) \quad \text{and} \quad 
{\bf z}^t = {\bf y}- A {\bf x}^t + {\bf c}^t.
\end{align} 
Here, ${\bf c}^t$ is a correction term that depends on the previous iterations, which significantly improves the convergence of AMP; ${\bf x}^t$ is the (block) sparse estimate at iteration $t$, and $\eta$ is a nonlinear function that imposes (block) sparsity. In particular,  if $\eta(\cdot)$ is the block soft-thresholding function as in \eqref{eqn:softThresh}, then AMP is equivalent to group LASSO in the asymptotic setting~\cite{CAMP,DONOHOGS} (see \cite{IntroAMP} for the details). 
   
 One of the most appealing features of the AMP is that its operation can be viewed as a denoising problem at each iteration. That is, when $N \to \infty$, ${\bf x}^t + A^* {\bf z}^t$ can be modeled as the sparse signal ${\bf x}$ plus zero-mean Gaussian noise, which is independent of the signal. This feature enables one to analytically predict the performance of AMP through a framework called \emph{state evolution} (SE). Concretely, if the mean-square error~(MSE) of AMP at iteration $t$ is denoted by 
 ${\rm MSE}^t$, then
 \begin{align}
 {\rm MSE}^{t+1} = \frac{1}{\delta B} \mathbb{E} \Big\{\big\|{\bf x}_B - \eta^{t}({ {\bf x}_B+ \sqrt{\rm MSE^t} {\bf z}_B}) \big\|_2^2 \Big\},
 \label{eqn:se}
 \end{align}
 where ${\bf z}_B \sim N(0,I_{B})$ and the distribution of ${\bf x}_B$ is the same as the empirical distribution of the blocks of the original vector ${\bf x}$. The expectation $\mathbb{E}\{\cdot\}$ is taken with respect to the vectors ${\bf z}_B$ and ${\bf x}_B$. 

\subsection{Phase transition}\label{sec:pt}

The performance of CS recovery algorithms can be characterized accurately by their \emph{phase transition} (PT) behavior.  Specifically, we define a two-dimensional phase space $(\delta,\rho) \in [0,1]$ that is partitioned into two regions: ``success" and ``failure", with these regions separated by the PT curve $(\delta,\rho(\delta))$. For the same value of $\delta$, algorithms with higher PT outperform algorithms with lower PT, i.e., guarantee the  exact recovery for more nonzero entries~$k$.
%
%

\section{Main results}


The thresholding function $\eta^t$ determines the performance of AMP. Indeed, different choices of $\eta^t$ may lead to fundamentally different performance. It has been shown in \cite{DONOHOGS, CAMP} that if $\eta^\text{soft}$ from \eqref{eqn:softThresh} is used, then the performance of AMP is equivalent to that of group LASSO. Since $\eta^\text{soft}$ is not necessarily the optimal thresholding function for group sparse signals, finding the optimal  function is of significant practical interest. 

In this paper, we characterize the optimal choice of the thresholding function $\eta^t$ in the strong undersampling regime, i.e., for $\delta \rightarrow 0$. Before we proceed, let us define optimality. Suppose that each block ${\bf x}_B$ is drawn independently from the distribution $F(\mathbf{x}_B)$ as defined in Section \ref{sec:notation}. We furthermore assume that $\eta^t$ is applied to each block separately. 
As is evident from \eqref{eqn:se}, each iteration of AMP is equivalent to a denoising problem, where the noise variance is equal to the MSE of the previous iteration. Therefore for a given initialization point, the effect of the thresholding function $\eta^t$ on AMP can be characterized by a discrete set $\{\text{MSE}^t\}_{t = 1}^{\infty}$. Thus, instead of considering a sequence of iteration-dependent thresholding functions $\{\eta^t\}_{t=1}^{\infty}$, we can consider a sequence of thresholding functions~$\overline{\eta}$ that depend on $\text{MSE}^t$ where $\overline{\eta}: \mathbb{R}^B \times \mathbb{R} \rightarrow \mathbb{R}^B$ with the propery ${\eta}^t({\bf y}_B) = \overline{\eta}({\bf y}_B,\text{MSE}^t)$. Since the MSE sequence is dependent on the initialization point of AMP (which can be chosen arbitrarily), we wish to optimize $\overline{\eta}$ with respect to \emph{all possible} initializations. Hence at each iteration, the problem is simplified to finding the optimal~$\overline{\eta}$  which is a function of $\mathbf{y}_B$ and any MSE value greater than zero. Now, suppose that the PT of AMP with $\bar{\eta}$ is given by~$\rho^{\overline{\eta}}\!(\delta, G)$. Then, we are interested in thresholding functions that achieve:
\begin{equation}
\rho^*(\delta) \triangleq \sup_{\overline{\eta}} \inf_G \rho^{\overline{\eta}}\!(\delta, G). 
\label{eqn:optimalpt}
\end{equation}
Such an $\overline{\eta}$ provides the best PT performance for the least favorable distribution---a reasonable assumption, since the distribution~$G$ is typically unknown in many practical applications.
Our first contribution characterizes the behavior of $\rho^*(\delta)$ for strong undersampling, i.e., for small values of $\delta$.
\begin{theorem}\label{thm:optimal}
The optimal PT $\rho^*(\delta)$ of AMP follows the behavior
\[
\rho^*(\delta) \sim \frac{B}{2 \log(1/\delta)} \quad \text{as} \quad \delta \rightarrow 0. 
\]
\end{theorem} 
As shown in Section \ref{sec:pfAMP}, this behavior is determined when $G$ is uniformly distributed on a sphere with infinite radius. We note that if the distribution $G$ is unknown, this theorem does not provide any guidelines on how to choose $\eta^t$ in \eqref{eqn:AMPiter}. The next theorem shows that group LASSO follows exactly the same behavior, and hence, block soft thresholding is the optimal choice for $\eta^t$ in the strong undersampling regime. 
\begin{theorem}\label{thm:LASSO} The PT $\rho^L(\delta)$ of group LASSO follows the behavior 
\[
\rho^\text{L} (\delta) \sim \frac{B}{2 \log(1/\delta)} \ \  {\rm as} \ \ \delta \rightarrow 0. 
\]
\end{theorem}
Combining Theorems \ref{thm:optimal} and \ref{thm:LASSO} reveals that for a fixed $B$ and in the strong undersampling regime, i.e., for $\delta \to 0$, the best achievable PT of AMP coincides with the phase transition of group LASSO. This result has two striking implications for strong undersampling and fixed block sizes: (i) Block soft thresholding is optimal and (ii) there is no thresholding function for AMP that outperforms group LASSO. 
Consequently, for decreasing $\delta$, AMP equipped with a better thresholding operator (than block soft-thresholding) such as Stein's shrinkage requires larger block sizes to outperform group LASSO in this regime.
It is worth mentioning that these results do not contradict those of \cite[Section 3.2]{DONOHOGS}, which show that AMP with Stein's shrinkage outperforms group LASSO for large block sizes. In fact, combining our results with those in  \cite{DONOHOGS} provides a better picture of the potential benefits of Stein's shrinkage within AMP.

\section{Proofs of the main results}

We outline the proofs of Theorems \ref{thm:LASSO} and~\ref{thm:optimal}. Since the proof of Theorem 3.1 requires the result of Theorem 3.2, we begin by proving the latter. 
\subsection{Proof of Theorem \ref{thm:LASSO}}
We start by deriving an implicit formula for the PT of group LASSO. We further use this formula and Laplace's method to obtain the behavior of the phase transition in the strong undersampling regime. 
\subsubsection{Implicit formula for the phase transition of group LASSO} 
Since the performance of group-LASSO has been shown to be equivalent to AMP with block soft thresholding, we can use the state evolution formalism to obtain the phase transition for group LASSO. Due to the properties of SE \eqref{eqn:se}, for any MSE value and $\rho(\delta)$ below the PT, the following holds:
\begin{equation} \label{eq:PTLASSO}
 \frac{1}{\delta B} \mathbb{E} \Big\{\big\|{\bf x}_B - \eta^\text{soft}({ {\bf x}_B+ \sqrt{\rm MSE} {\bf z}_B} ; \tau \sqrt{\rm MSE}) \big\|_2^2 \Big\} \leq {\rm MSE}.
\end{equation} 
To ensure that (\ref{eq:PTLASSO}) is satisfied while achieving the optimal phase transition with respect to $\tau$, we use the minimax MSE $M_{B}$, which corresponds to \cite{DONOHOGS}
\begin{equation}
M_B = \frac{1}{B}\inf_{\tau} \sup_{G} \mathbb{E} \Big\{\left\|{\bf x}_B - \eta^\text{soft}({\bf x}_B+ {\bf z}_B;\tau) \right\|_2^2 \Big\}
\label{minimaxBlock}
\end{equation} 
in the asymptotic setting. 
Donoho \emph{et al.} showed in \cite{DONOHOGS}:
\begin{eqnarray*}
{M_B} = \frac{1}{B}\inf_{\tau} \Big \{\epsilon(B+{\tau}^2) + (1-\epsilon) \int_{\tau^2}^{\infty} (\sqrt{x} - \tau)^2 \frac{1}{2^{\frac{B}{2}}\Gamma(\frac{B}{2})} x^{\frac{B}{2}-1} e^{-\frac{x}{2}} \mathrm{d}x \Big \}. 
\end{eqnarray*}
One can rigorously show that the optimal phase transition obeys $M_B = \delta$ \cite{IntroAMP}. Using simple calculus, we obtain, $\rho^\text{L}(\delta)$ as follows.   
\begin{lemma}
\label{lemma:lassoPT}
The phase transition for group LASSO is given by: 
\begin{align}
\rho^{L}(\delta) = \frac{B\delta -\int_{{{\tau^*}}^2}^{\infty} (\sqrt{x}-{\tau^*})^2f(x)\mathrm{d}x}{\delta(B+{{\tau^*}}^2 - \int_{{\tau^*}^2}^{\infty} (\sqrt{x}-{\tau^*})^2 f(x) \mathrm{d}x)},
\label{eqn:rhodelta}
\end{align}
where $f(x)$ is the probability density function of $\Gamma(\frac{B}{2},\frac{1}{2})$ and the optimal threshold parameter, ${\tau^*}$, satisfies
\begin{equation}
\delta = \frac{ -(B+{\tau^*}^2) \int_{{\tau^*}^2}^{\infty} ({\tau^*} - \sqrt{x}) f(x) \mathrm{d}x  + \tau^* \int_{{\tau^*}^2}^{\infty} (\sqrt{x} - \tau^*)^2 f(x) \mathrm{d}x}{B\tau^{*} - B \int_{{\tau^*}^2}^{\infty} ({\tau^*}-\sqrt{x})f(x)\mathrm{d}x}. \\
\label{eqn:deltatau}
\end{equation}
\end{lemma}
It is important to note that $\rho^{L}(\delta)$ is independent of distribution $G$ (proof is omitted in this paper as it is a mere extension of the one provided in \cite{CAMP}).

\subsubsection{Phase transition behavior for $\delta \to 0$}
We are interested in observing the behavior of group LASSO for $\delta \to 0$. Intuitively, for such regime, a very sparse signal is recovered since $\epsilon \to 0$ ($\epsilon = \rho\delta$). From the definition of block soft thresholding, $\tau^*$ must be large to promote sparse recovery. Using this knowledge, the integrals in~(\ref{eqn:deltatau}) and ~(\ref{eqn:rhodelta}) can be approximated via Laplace's method. One such integral is
\begin{eqnarray*}
I_1 = \int_{{\tau^*}^2}^{\infty} ({\tau^*}-\sqrt{x})x^{\frac{B}{2}-1} e^{-\frac{x}{2}} \mathrm{d}x. 
\end{eqnarray*}
We begin by letting $y = -{\tau^*}+\sqrt{x}$. As a result, the expression for $I_1$ is simplified to
\begin{eqnarray*}
I_1 = -2e^{-\frac{{{\tau^*}}^2}{2}} \int_{0}^{\infty} y(y+{\tau^*})^{B-1} e^{-\frac{y^2}{2}} e^{-{\tau^*}y} \mathrm{d}y.
\end{eqnarray*}
Suppose we break the integral into two parts:
\begin{eqnarray*}
I_1 = -2e^{-\frac{{\tau^*}^2}{2}} \int_{0}^{\gamma} y(y+{\tau^*})^{B-1} e^{-\frac{y^2}{2}} e^{-{\tau^*}y} \mathrm{d}y 
- 2e^{-\frac{{{\tau^*}}^2}{2}} \int_{\gamma}^{\infty} y(y+{\tau^*})^{B-1} e^{-\frac{y^2}{2}} e^{-{\tau^*}y} \mathrm{d}y,
\end{eqnarray*}
where $\gamma$ is close to zero. Due to the decaying exponential in the integrand, the second integral, which is of the order $O(e^{-\frac{{\tau^*}^2}{2}}e^{-{\tau^*}\gamma})$, is much smaller than the first integral. Thus, we approximate $I_1$ by the first integral. Denote $g(y) = y(y+{\tau^*})^{B-1} e^{-\frac{y^2}{2}}$. Since~$\gamma$ is small, $g(y)$ can be well approximated with a first order taylor expansion around $y = 0$. This yields 
\begin{eqnarray*}
I_1 = -2e^{-\frac{{\tau^*}^2}{2}} \left(\int_0^{\gamma} y({\tau^*})^{B-1} e^{-{{\tau^*}}y} \mathrm{d}y + O({\tau^*}^{B-5})\right) \sim -2e^{-\frac{{\tau^*}^2}{2}} (\tau^*)^{B-3} ,
\end{eqnarray*}
where the second approximation is obtained via integration by parts. The other integral $I_2$ in (\ref{eqn:deltatau}) that we wish to approximate is given by
\begin{eqnarray*}
I_2 = \int_{{\tau^*}^2}^{\infty} (-{\tau^*}+\sqrt{x})^2x^{\frac{B}{2}-1} e^{-\frac{x}{2}} \mathrm{d}x. 
\end{eqnarray*}
Using similar techniques, we find $I_2 \sim 4e^{-\frac{{\tau^*}^2}{2}}(\tau^*)^{B-4}$. Substituting the approximations of $I_1$ and $I_2$ in (\ref{eqn:rhodelta}), we note that the numerator is dominated by the term $B\delta$, and the denominator by $\delta{\tau^*}^2$. With this behavior, we obtain
\begin{equation}
\rho^L(\delta) \sim \frac{B}{{\tau^*}^2},
\label{eqn:rhoapprox}
\end{equation}
as $\tau^* \to \infty$. The expression for $\delta$ is similarly derived to be
\begin{eqnarray*}
\delta \sim \frac{2(B+{\tau^*}^2)e^{-\frac{{\tau^*}^2}{2}}({\tau^*})^{B-3} + 4\tau^*e^{-\frac{{\tau^*}^2}{2}}(\tau^*)^{B-4}}{B\tau^*h(B) + 2Bh(B)e^{-\frac{{\tau^*}^2}{2}} ({\tau}^*)^{B-3}},
\end{eqnarray*}
where $h(B) = 2^{\frac{B}{2}} \Gamma(\frac{B}{2})$. Since $\tau^* \to \infty$, 
\begin{eqnarray}
\delta \sim \frac{{\tau^*}^{B-2} e^{-\frac{{\tau^*}^2}{2}}}{{2^{\frac{B}{2}-1}} B\Gamma(\frac{B}{2})}.
\label{eqn:deltaapprox}
\end{eqnarray}
From (\ref{eqn:deltaapprox}), it is clear that for large $\tau^*$, $\log(\delta^{-1}) \sim  {\tau^*}^2/2$. Combining this with (\ref{eqn:rhoapprox}) yields the desired result.

\subsection{Proof of Theorem \ref{thm:optimal}} \label{sec:pfAMP}
We now derive an expression for $\rho^*(\delta)$ when \mbox{$\delta \to 0$}. By definition,
\begin{eqnarray*}
M_B \triangleq \frac{1}{B}\inf_{\eta} \sup_{G} \mathbb{E} \Big\{\left\|{\bf x}_B - \eta({\bf x}_B+ {\bf z}_B) \right\|_2^2\Big \} = \frac{1}{B} \sup_{G} \inf_{\eta} \mathbb{E} \Big\{\left\|{\bf x}_B - \eta({\bf x}_B+ {\bf z}_B) \right\|_2^2\Big \},
\end{eqnarray*}
where $\mathbf{x}_B$ drawn from $F({\bf x}_B)$ as defined in  Section \ref{sec:notation}. Defining $\mathbf{y}_B= {\bf x}_B+ {\bf z}_B$, the Bayes estimator, $\mathbb{E}[{\bf x}_B|{\bf y}_B] $ minimizes the risk for every $G$. Thus
\begin{eqnarray*}
M_B(G) \triangleq  \frac{1}{B} \mathbb{E} \Big\{\left\|{\bf x}_B - \mathbb{E}[{\bf x}_B|{\bf y}_B] \right\|_2^2 \Big\}.
\end{eqnarray*}
Consider $F({\bf x}_B) = (1- \epsilon) \delta_{0}(\|{\bf x}_B\|_2) + \epsilon G^*({\bf x}_B)$, where $G^*$ is a distribution that is uniform on a sphere with radius $\mu$. We relate $\mu$ to $\epsilon$ as follows
\begin{eqnarray}
(1-\gamma)\log\Big(\frac{\epsilon}{1-\epsilon}\Big) = -\frac{\mu^2}{2} \hspace{.2in}
\gamma\log\Big(\frac{\epsilon}{1-\epsilon}\Big) = {-a\mu},
\label{eqn:conds}
\end{eqnarray}
where $\gamma$ is chosen such that as $\epsilon \to 0$, $\mu, a \to \infty$. One such choice could be 
\begin{eqnarray*}
\delta \sim \log\left(\frac{1-\epsilon}{\epsilon}\right)^{-1/4}. 
\end{eqnarray*}
Before we proceed to find the risk function associated with $G^*$, we will develop an intuition on the behavior of the Bayes estimator. 

Define $G = \mu{\bf \theta}_B$, where ${\bf \theta}_B$ is a distribution that is uniform on a sphere with unit radius. Denote $\hat{{\bf x}}_B(G^*)$ as the Bayes estimator for $G^*$. The k-th element of the Bayes estimator, $\hat{{\bf x}}_{[B,k]}$, is approximately given by 
\begin{equation}
\hat{{\bf x}}_{[B,k]}(G^*) = \frac{\epsilon \int \mu {\theta}_{[B,k]} e^{-\frac{1}{2}\left\|{\bf y}_B -  \mu{\bf \theta}_B \right\|_2^2} \mathrm{d\sigma}(\theta_B)}{(1-\epsilon)e^{-\frac{\left|{\bf y}_B \right\|_2^2}{2}} +  \epsilon \int e^{-\frac{1}{2}\left\|{\bf y}_B - \mu{\bf \theta}_B \right\|_2^2} \mathrm{d\sigma}(\theta_B)},
\label{eqn:BayesEstimator}
\end{equation}
where $\mathrm{d\sigma}(\theta_B)$ denotes the Haar measure on the unit sphere. Using the conditions in (\ref{eqn:conds}), we can rewrite (\ref{eqn:BayesEstimator}) as follows:
\begin{equation}
\hat{{\bf x}}_{[B,k]}(G^*) = \frac{ \int \mu {\theta}_{[B,k]} e^{\mu \langle {\bf y}_B,{\bf \theta}_{B} \rangle -\mu^2-a\mu} \mathrm {d \sigma({\bf \theta}_B)}}{1 +  \int e^{ \mu \langle {\bf y}_{B},{\bf \theta}_{B} \rangle - \mu^2-a\mu} \mathrm {d \sigma({\bf \theta}_B)}}.
\label{eqn:estimator}
\end{equation}
We wish to find an approximation for $\hat{{\bf x}}_{[B,k]}$. To that end, we approximate the integrals in (\ref{eqn:estimator}) using Laplace's method. This is achieved by determining the ${\theta^*}_B$ that maximizes the exponent of the integrand under the constraint that $\left \| {\bf \theta}_B \right \|_2 = 1$. The method of Lagrange multipliers yields
\vspace{-.1in}
\begin{eqnarray*}
\textstyle \frac{d}{d{\bf \theta}_B} \Big[ \langle {\bf \theta}_B,{\bf y}_{B} \rangle - \gamma\big(\sum_{j = 1}^{B} {\bf \theta}_{[B,j]}^2-1\big) \Big]= 0.
\end{eqnarray*}
The exponents in both integrals are maximized if ${\theta^*}_B = \frac{{\bf y}_B}{\left\|{\bf y}_B\right\|_2}$. Thus, (\ref{eqn:estimator}) is approximated by
\begin{eqnarray}
\hat{{\bf x}}_{[B,k]}(G^*) \approx \frac{\mu {\bf \theta^*}_{[B,k]}e^{\mu \left \|{\bf y}_B \right\|_2 - \mu^2-a\mu}}{1 +  e^{\mu \left \|{\bf y}_B \right\|_2 - \mu^2-a\mu}} \approx
\left\{\begin{array}{ll}
0  &  \left \| {\bf y}_B \right \|_2 < \mu + a \\[0.1cm]
\frac{{\mu {\bf y}_{[B,k]}}}{\left\| {\bf y}_B \right \|_2} &  \left \| {\bf y}_B \right \|_2 \geq \mu + a
\end{array}\right.
\end{eqnarray} 
Thus, we see that  $\hat{{\bf x}}_{B}(G^*)$, is approximately given by, 
\begin{eqnarray}
\hat{{\bf x}}_{B}(G^*) \approx
\left\{\begin{array}{ll}
0  &  \left \| {\bf y}_B \right \|_2 < \mu + a \\[0.1cm]
\frac{{\mu {\bf y}_{B}}}{\left\| {\bf y}_B \right \|_2}  &  \left \| {\bf y}_B \right \|_2 \geq \mu + a.
\end{array}\right.
\label{eqn:cond1}
\end{eqnarray} 
Therefore, very interestingly, the Bayes estimator for this distribution behaves similiarly to a hard thresholding function. \\
In the following lemmas, we further characterize the Bayes estimator and the risk function associated with $G^*$. These results will lead us to the behavior of the phase transition for this distribution.

\begin{lemma}
The Bayes estimator associated with $G^*$, $\hat{{\bf x}}_B(G^*)$, behaves as $\left\|\hat{{\bf x}}_B(G^*)\right\|_2 \leq \mu$.
\label{lemma:ineqBay}
\end{lemma}
\vspace{-.1in}
\begin{proof}
Suppose that $\left\|\hat{{\bf x}}_B(G^*)\right\|_2 > \mu$. Define the estimator, $\hat{{\bf w}}_B$, as:
\begin{eqnarray*}
\hat{{\bf w}}_B = 
\left\{\begin{array}{ll}
 \hat{{\bf x}_B}(G^*) &  \left\|\hat{{\bf x}}_B(G^*)\right\|_2 \leq \mu \\[0.1cm]
\mu\frac{\hat{{\bf x}}_B(G^*)}{\left\|\hat{{\bf x}}_B(G^*)\right\|_2} &  \left\|\hat{{\bf x}}_B(G^*)\right\|_2 > \mu 
\end{array}\right.
\end{eqnarray*}

By definition, $\left\|\hat{{\bf w}}_B\right\|_2 \leq \left\|\hat{{\bf x}}_B(G^*)\right\|_2$. The risk associated with estimator $\hat{{\bf w}}_B$ is given by:
\begin{eqnarray*}
M_{\hat{{\bf w}}_B}(G^*) &=& (1-\epsilon)\mathbb{E}\left\|\hat{{\bf w}}_B\right\|_2^2 \\ &+& \epsilon\mathbb{E} \Big\{\left\|\hat{{\bf w}}_B - \mu\theta_B\right\|_2^2 \Big\|  \left\|\hat{{\bf x}}_B(G^*)\right\|_2 > \mu  \Big\} \mathbb{P}\Big\{ \left\|\hat{{\bf x}}_B(G^*)\right\|_2 > \mu\Big\} \\ &+& \epsilon\mathbb{E} \Big\{\left\|\hat{{\bf w}}_B - \mu\theta_B\right\|_2^2 \Big\|  \left\|\hat{{\bf x}}_B(G^*)\right\|_2 \leq \mu  \Big\} \mathbb{P}\Big\{ \left\|\hat{{\bf x}}_B(G^*)\right\|_2 \leq \mu\Big\}
\end{eqnarray*}
By geometric reasoning, it is easy to see that $M_{\hat{{\bf w}}_B}(G^*) < M_B(G^*)$ This however is a contradiction since the Bayes estimator, $\hat{{\bf x}}_B(G^*)$, is the optimal estimator for the distribution $G^*$. Thus, we conclude that $\left\|\hat{{\bf x}}_B\right\|_2 \leq \mu$.
\end{proof}
\\
Before proceeding to finding the risk associated with $G^*$, we provide the following useful lemma. 
\begin{lemma}\label{lemma:prob1}
Let ${\bf y}_B = \mu{\theta}_B + {\bf z}_B$. If $\mu$ and $a$ satisfy (\ref{eqn:conds}), then
\begin{eqnarray*}
\mathbb{P}\Big\{\left\|{\bf y_B} \right\|_2 > \mu + a\Big\} \leq 2e^{-\frac{a^2}{2}} + e^{-(\frac{1}{2}-\frac{B}{2a^2})a^2} \Big(\frac{B}{a^2}\Big)^{-\frac{B}{2}}.
\end{eqnarray*}
\end{lemma}

\begin{proof}
Since $\mathbf{z}_B \sim N(0,I_B)$, $\|\mathbf{z}_B\|$ and $\frac{\mathbf{z}_B}{\|\mathbf{z}_B\|_2}$ are independent. Furthermore,  $\frac{\mathbf{z}_B}{\|\mathbf{z}_B\|_2}$ has a uniform (Haar) distribution on the unit sphere in $\mathbb{R}^B$. Therefore, $\mathbb{P}\Big\{\left\|{\bf y_B} \right\|_2 > \mu + a\Big\}$ does not depend on $\theta_B$ and hence we set $\theta_B =(1, 0, \ldots, 0)$. 
\begin{eqnarray}\label{eq:maintermprob}
\mathbb{P}\Big\{\left\|{\bf y_B} \right \|_2 > \mu + a\Big\} &=& \mathbb{P}\Big\{\left\|(\mu,0,\dots,0) +  {\bf z}_B\right\|_2^2 \geq (\mu+a)^2\Big\} \nonumber \\ 
& =& \mathbb{P}\Big\{\left\|{\bf z}_B \right\|_2^2 + 2{\mu}{\bf z}_{[B,1]}\geq a^2 + 2a\mu\Big\} \nonumber \\
&=& \mathbb{P}\Big\{ \left\|{\bf z}_B \right\|_2^2 + 2{\mu}{\bf z}_{[B,1]}\geq a^2 + 2a\mu \Big | \left|{\bf z}_{[B,1]}\right| > a\Big\}\mathbb{P}\Big\{\left|{\bf z}_{[B,1]}\right| > a\Big\}\nonumber \\ &+& {\mathbb{P}}\Big\{\left\|{\bf z}_B \right\|_2^2 + 2{\mu}{\bf z}_{[B,1]}\geq a^2 + 2a\mu \Big| \left|{\bf z}_{[B,1]}\right| < a\Big\}\mathbb{P}\Big\{\left|{\bf z}_{[B,1]}\right| < a\Big\} \nonumber \\
&\leq&\mathbb{P}\Big\{\left|{\bf z}_{[B,1]}\right| > a\Big\}  \nonumber \\
 &+& {\mathbb{P}}\Big\{\left\|{\bf z}_B \right\|_2^2 + 2{\mu}{\bf z}_{[B,1]}\geq a^2 + 2a\mu \Big| \left|{\bf z}_{[B,1]}\right| < a\Big\} \mathbb{P}\Big\{\left|{\bf z}_{[B,1]}\right| < a\Big\}. \nonumber \\
\end{eqnarray}
Using standard bounds on the tail of a Gaussian random variable we obtain,
\begin{eqnarray}
\mathbb{P}\Big\{\left|{\bf z}_{[B,1]}\right| > a\Big\} \leq 2e^{-\frac{a^2}{2}}
\label{eqn:probab1}
\end{eqnarray}
 Furthermore,  
\begin{align}
\qquad\qquad&{\mathbb{P}}\Big\{\left\|{\bf z}_B \right\|_2^2 + 2{\mu}{\bf z}_{[B,1]}\geq a^2 + 2a\mu \Big| \left|{\bf z}_{[B,1]}\right| < a\Big\}\mathbb{P}\Big\{\left|{\bf z}_{[B,1]}\right| < a\Big\} \nonumber \\
&\qquad\qquad\leq \mathbb{P}\Big\{ \left\|{\bf z}_B \right\|_2^2 \geq a^2 \Big| |{\bf z}_{[B,1]}| < a\Big\}\mathbb{P}\Big\{\left|{\bf z}_{[B,1]}\right| < a\Big\}  \nonumber \\  
&\qquad\qquad\leq \mathbb{P}\Big\{\left\|{\bf z}_B \right\|_2^2 \geq a^2 \Big\}  \nonumber \\ 
&\qquad\qquad\leq e^{-(\frac{1}{2}-\frac{B}{2a^2})a^2} \Big(\frac{B}{a^2}\Big)^{-\frac{B}{2}} \label{eqn:probab2}
\end{align}
The technique to obtain the last inequality can be found in ~\cite[Section 5]{ARSH}. Combining \eqref{eq:maintermprob}, (\ref{eqn:probab1}), and (\ref{eqn:probab2}), we obtain the desired result. 
\end{proof}

In the following lemma, we characterize the risk function associated with $G^*$.
\begin{lemma}
Suppose ${\bf x}_B \sim F({\bf x}_B)$. Denote the associated risk function as $M_B({G^*})$. Let $\gamma > 0$, where $\gamma$ can be made arbitrarily small. Set $(1-\gamma)\log \big(\frac{1-\epsilon}{\epsilon} \big) = \frac{1}{2}\mu^2$.
Then,
\begin{eqnarray*}
M_B(G^*) \sim \epsilon\frac{2(1-\gamma)}{B} \log \Big(\frac{1-\epsilon}{\epsilon} \Big) \hspace{0.1in} \rm{as}\hspace{0.05in} \epsilon \to 0. 
\end{eqnarray*}
\end{lemma}

\begin{proof} 
By definition, the risk is equal to
\begin{eqnarray*}
M_B(G^*) = \underbrace{\frac{\epsilon}{B} \mathbb{E} \Big \{\left\|\mu{\bf \theta}_B - \mathbb{E}[{\bf x}_B|{\bf y}_B] \right\|_2^2 \Big\}}_{M_B^1(G^*)}+ \underbrace{\frac{1-\epsilon}{B} \mathbb{E} \Big\{\left\|0 - \mathbb{E}[{\bf x}_B|{\bf z}_B] \right\|_2^2 \Big\}}_{M_B^2(G^*)},
\end{eqnarray*}
where ${\bf y}_B = \mu{\bf \theta}_B + {\bf z}_B$. The expectation of $M_B^1(G^*)$ is with respect to ${\bf z}_B \sim N(0,I_{B})$ as well as ${\bf x}_B \sim G^*$. The expectation of $M_B^2(G^*)$ is with respect to $z_B$. We begin by computing $M_B^1(G^*)$. Applying the conditions in (\ref{eqn:cond1}), $M_B^1(G^*)$ can be expressed as
\begin{align}
M_B^1(G^*) = & \underbrace{\frac{\epsilon}{B} \mathbb{E} \Big\{\left\|\mu{\bf \theta}_B - \mathbb{E}[{\bf x}_B|{\bf y}_B]\right\|_2^2 \Big|\left\|{\bf y}_B \right\|_2 < \mu + a - \alpha \Big\} \mathbb{P}\Big\{\left\|{\bf y}_B \right\|_2 < \mu + a - \alpha\Big \}}_{M_B^{1*}(G^*)} \nonumber \\& +  \underbrace{\frac{\epsilon}{B} \mathbb{E} \Big\{\left\|{\bf \mu \theta}_B - \mathbb{E}[{\bf x}_B|{\bf y}_B] \right\|_2^2 \Big|\left\|{\bf y}_B \right\|_2 \geq \mu + a - \alpha\Big\} \mathbb{P}\Big\{\left\|{\bf y_B} \right \|_2 \geq \mu + a - \alpha\Big\}}_{M_B^{1**}(G^*)} . 
\label{eqn:M1term}
\end{align}
where $\alpha$ is a constant that can be made arbitrarily small. We first begin by bounding ${M_B^{1**}(G^*)}$. Using Lemma \ref{lemma:ineqBay}, we note
\begin{eqnarray*}
\left\|{\bf \mu \theta}_B - \mathbb{E}[{\bf x}_B|{\bf y}_B] \right\|_2  \leq \left\|{\bf \mu \theta}_B \right \|_2 + \left\|\mathbb{E}[{\bf x}_B|{\bf y}_B] \right \|_2 \leq 2\mu.
\end{eqnarray*} 
Therefore, 
\begin{align}
{M_B^{1**}(G^*)} \leq \frac{4\mu^2\epsilon}{B}\mathbb{P}\Big\{\left\|{\bf y_B} \right \|_2 \geq \mu + a - \alpha\Big\} 
\label{eqn:M11}
\end{align}
Next,  we consider ${M_B^{1*}}$. Recall that the k-th element of Bayes estimator $\mathbb{E}[{\bf x}_B|{\bf y}_B]$ associated with the distribution $G^*$ is given by (\ref{eqn:BayesEstimator}). Then, given $\left\|{\bf y}_B \right \|_2 \leq \mu + a - \alpha$, we have

\begin{align*}
|\hat{{\bf x}}_{[B,k]}(G^*)| =  \frac{\Big|\int \mu {{\theta}}_{[B,k]}^{'} e^{\mu \langle {\bf y}_B,{\bf \theta}_{B}^{'} \rangle -\mu^2-a\mu} \mathrm {d \sigma({\bf \theta}_B^{'})}\Big|}{\Big|1 +  \int e^{ \mu \langle {\bf y}_{B},{\bf {\theta}}_{B}^{'} \rangle - \mu^2-a\mu} \mathrm {d \sigma({\bf {\theta}}_B^{'})}\Big|} \leq   {\int \mu \Big|{{\theta}}_{[B,k]}^{'} e^{\mu \langle {\bf y}_B,{\bf \theta}_{B}^{'} \rangle -\mu^2-a\mu} \Big|\mathrm {d \sigma({\bf \theta}_B^{'})}}
\end{align*}
To provide a further bound, the exponent can be maximized by letting $\theta_B^{'} = \frac{{\bf y}_B}{\left\|{\bf y}_B \right\|_2}$. Moreover, $|{\theta_{[B,k]}^{'}}|< \mu$. Thus,
\begin{align*}
\mu\left|{\theta}_{[B,k]}\right|- \mu^2{e^{-\alpha\mu}} \leq \left|\mu{\bf \theta}_{[B,k]} - \mathbb{E}[{\bf x}_{[B,k]}|{\bf y}_{[B,k]}] \right|\leq \mu\left|{\theta}_{[B, k]}\right|+ \mu^2{e^{-\alpha\mu}}.
\end{align*}
Since $\alpha$ can be chosen arbitrarily, we let $\alpha \to 0$ in a manner such that $\alpha\mu \to \infty$. Therefore,
\begin{eqnarray}
M_B^{1*}(G^*) &=&\frac{\epsilon}{B} \mathbb{E} \Big\{\left\|\mu{\bf \theta}_B - \mathbb{E}[{\bf x}_B|{\bf y}_B]\right\|_2^2 \Big|\left\|{\bf y}_B \right\|_2 < \mu + a - \alpha \Big\} \mathbb{P}\Big\{\left\|{\bf y}_B \right\|_2 < \mu + a - \alpha\Big \} \nonumber \\  &=& \frac{\epsilon}{B} \mathbb{E} \Big\{\left\|\mu{\bf \theta}_B \right\|_2^2 \Big\} \mathbb{P}\Big\{\left\|{\bf y}_B \right\|_2 < \mu + a - \alpha\Big \} \nonumber \\
&=&\frac{\epsilon\mu^2}{B}\mathbb{P}\Big\{\left\|{\bf y}_B \right\|_2 < \mu + a - \alpha\Big\}
\label{eqn:M12}
\end{eqnarray}

Combining (\ref{eqn:M1term}),(\ref{eqn:M11}), (\ref{eqn:M12}), and the fact that $M_B^1(G^*) = M_B^{1*}(G^*) + M_B^{1**}(G^*)$ yields
\begin{eqnarray*}
M_B^1(G^*) \leq \frac{\epsilon\mu^2}{B}\mathbb{P}\Big\{\left\|{\bf y}_B \right\|_2 < \mu + a - \alpha\Big\} + \frac{4\mu^2\epsilon}{B}\mathbb{P}\Big\{\left\|{\bf y_B} \right \|_2 \geq \mu + a - \alpha\Big\}.  
\end{eqnarray*}
After further manipulation, we obtain
\begin{eqnarray*}
 \left|M_B^1(G^*)-  \frac{\epsilon \mu^2}{B}\right| &\leq& \frac{3 \epsilon \mu^2}{B} \mathbb{P}\Big\{\left\|{\bf y}_B \right \|_2 \geq \mu + a - \alpha\Big\} \nonumber \\
&\leq&  \frac{3 \epsilon \mu^2}{B} \left(2e^{-\frac{(a-\alpha)^2}{2}} + e^{-(\frac{1}{2}-\frac{B}{2(a-\alpha)^2})(a-\alpha)^2} \Big(\frac{B}{a-\alpha}\Big)^{-B}B^{\frac{B}{2}} \right).
\end{eqnarray*}

Since $\alpha\to 0$, and $a \to \infty$, we conclude that $M_B^1(G^*) \sim \frac{\epsilon\mu^2}{B}$.
We next characterize $M_B^2(G^*)$:
\begin{eqnarray*}
M_B^2(G^*) &=& \frac{1- \epsilon}{B} \mathbb{E} \Big \{ \Big \|\mathbb{E}[\mathbf{x}_B | \mathbf{z}_B] \Big \|_2^2 \Big \} \nonumber \\
&=&\frac{1- \epsilon}{B} \mathbb{E} \Big \{ \Big\|\mathbb{E}[\mathbf{x}_B | \mathbf{z}_B] \Big \|_2^2  \Big| \|\mathbf{z}_B\|_2 \geq \mu+a -\alpha \Big \} \mathbb{P} \Big\{\|\mathbf{z}_B\|_2 \geq \mu+a-\alpha \Big\} \nonumber \\ &+& \frac{1- \epsilon}{B} \mathbb{E} \Big \{ \Big\|\mathbb{E}[\mathbf{x}_B | \mathbf{z}_B] \Big \|_2^2  \Big| \|\mathbf{z}_B\|_2 < \mu+a - \alpha \Big \} \mathbb{P} \Big\{\|\mathbf{z}_B\|_2 < \mu+a - \alpha \Big\}  \nonumber \\
&\leq&\frac{(1- \epsilon)\mu^2}{B} \mathbb{P} \Big\{\|\mathbf{z}_B\|_2 \geq \mu+a-\alpha\Big\} + \frac{1- \epsilon}{B}\Big({e^{-\alpha\mu}}\Big)^2 
\end{eqnarray*}
It is straightforward to see that this term is also negligible compared to $\frac{\epsilon \mu^2}{B}$ and therefore $M_B(G^*) = M_B^1(G^*)+ M_B^2(G^*) \sim \frac{\epsilon \mu^2}{B}$.
Combining this with (\ref{eqn:conds}) gives the desired result. 
\end{proof}\\
The PT for the distribution $F({\bf x}_B)$ is found by letting $M_B(G^*) = \delta$ to obtain $\rho(\delta) \sim {B}/({2(1-\gamma)\log(\frac{1-\epsilon}{\epsilon})})$. Since $\epsilon \to 0$ and $\gamma$ can be made arbitrarily small, we have $\rho(\delta) \lesssim {B}/({2\log(\frac{1}{\delta})})$. From (\ref{eqn:optimalpt}), we know that $ \rho^*(\delta) \lesssim \rho(\delta) \lesssim {B}/({2\log(\frac{1}{\delta})})$.  Moreover, recall from Theorem \ref{thm:LASSO}, $\rho^\text{L}(\delta) \sim B/(2\log(\frac{1}{\delta}))$ and is independent of the distribution $G$. By definition of optimality, $\rho^*(\delta) \gtrsim \rho^\text{L}(\delta) \sim B/(2\log(\frac{1}{\delta})) $. Comparing this with the upper bound for $\rho^*(\delta)$, we conclude that $\rho^*(\delta) \sim B/(2\log(\frac{1}{\delta}))$.

\bibliographystyle{unsrt}
\bibliography{references2}

\end{document}